\documentclass[1p]{elsarticle}

\pdfoutput=1

\usepackage[T1]{fontenc}
\usepackage[latin9]{inputenc}
\usepackage{verbatim}
\usepackage{float}
\usepackage{amsthm}
\usepackage{amsmath}
\usepackage{amssymb}
\usepackage{graphicx}
\usepackage{color}
\usepackage{url}
\usepackage{natbib}

\definecolor{mypink}{rgb}{0.858, 0.188, 0.478}
\newcommand{\rev}[1]{{\color{black}{#1}}}

\makeatletter

\floatstyle{ruled}
\newfloat{algorithm}{tbp}{loa}
\providecommand{\algorithmname}{Algorithm}
\floatname{algorithm}{\protect\algorithmname}

\numberwithin{equation}{section}
\numberwithin{figure}{section}
\theoremstyle{plain}
\newtheorem{thm}{\protect\theoremname}[section]
\theoremstyle{definition}
\newtheorem{defn}[thm]{\protect\definitionname}
\theoremstyle{remark}

\theoremstyle{plain}
\newtheorem{lem}[thm]{\protect\lemmaname}

\newtheorem*{lem*}{Lemma}

\theoremstyle{remark}
\newtheorem{rem}[thm]{\protect\remarkname}
\theoremstyle{plain}

\theoremstyle{plain}
\newtheorem{proposition}[thm]{\protect\propositionname}

\providecommand{\claimname}{Claim}
\providecommand{\definitionname}{Definition}
\providecommand{\lemmaname}{Lemma}
\providecommand{\remarkname}{Remark}
\providecommand{\theoremname}{Theorem}
\providecommand{\corollaryname}{Corollary}
\providecommand{\propositionname}{Proposition}

\newcommand{\CN}{\mathbb{C}^N}

\newcommand{\hx}{\widehat{x}}
\newcommand{\hy}{\widehat{y}}
\newcommand{\I}{\iota}

\begin{document}


\title{ On Signal Reconstruction from FROG Measurements}


\author[add1]{Tamir Bendory}
\ead{tamir.bendory@princeton.edu}
\author[add2]{Dan Edidin}
\ead{edidind@missouri.edu}
\author[add3]{Yonina C. Eldar}
\ead{yonina@ee.technion.ac.il}

\address[add1]{The Program in Applied and Computational Mathematics,
	Princeton University, Princeton, NJ, USA}
\address[add2]{
	Department of Mathematics, University of Missouri, Columbia, Missouri, USA}
\address[add3]{The Andrew and Erna Viterbi Faculty of Electrical Engineering,
	Technion -- Israel Institute of Technology, Haifa, Israel}

{\let\thefootnote\relax\footnote{D. Edidin acknowledges support from Simons Collaboration Grant 315460. Y.C. Eldar acknowledges support from 
 the European Union's Horizon
2020 research and innovation program under grant agreement No. 646804--ERCCOG--BNYQ, and from the Israel Science Foundation under Grant no.
335/14.}}


\begin{abstract}
Phase retrieval refers to recovering a signal from its Fourier magnitude. This problem arises naturally in many scientific applications, such as ultra-short laser pulse characterization and diffraction imaging.  Unfortunately,  phase retrieval is  ill-posed for almost all one-dimensional signals.   
 In order to characterize a laser pulse and overcome the ill-posedness, it is common to use a technique called Frequency-Resolved Optical
 Gating (FROG). In FROG, the measured data, referred to as FROG trace, is the Fourier magnitude of the product of the underlying signal with several translated versions of itself. The FROG trace results in a system of phaseless  quartic Fourier measurements.
In this paper, we prove that it suffices to consider only three translations of the signal to determine almost all bandlimited signals, up to trivial ambiguities.
 In practice, one usually also has access to the signal's Fourier magnitude. We show that in this case  only two translations suffice.  Our results significantly improve upon earlier work. 
\end{abstract}

\maketitle

\section{Introduction} \label{sec:introduction}

\emph{Phase retrieval} is the problem of  estimating a signal from its Fourier magnitude. This problem plays a key role in many scientific and engineering applications, among them X-ray crystallography, speech recognition, blind channel estimation, alignment tasks and astronomy \cite{millane1990phase,elser2017benchmark, rabiner1993fundamentals,baykal2004blind,fienup1987phase,bendory2017bispectrum}. Optical applications are of particular interest since optical devices, such as a charge-coupled device (CCD) and the human eye, cannot detect  phase information of the light wave~\cite{shechtman2015phase}.

Almost all one-dimensional signals cannot be determined uniquely from their Fourier magnitude. 
This immanent ill-posedness  makes this problem substantially more challenging than its  multi-dimensional counterpart, which is well-posed for almost all signals~\cite{bendory2017fourier}.
Two exceptions for one-dimensional signals that are determined uniquely from their Fourier magnitude are minimum phase signals and sparse signals with non-periodic support~\cite{huang2016phase,jaganathan2013sparse}\rev{; see also~\cite{beinert2015ambiguities}}. One popular way to overcome the non-uniqueness is by collecting additional information on the sought signal beyond its Fourier magnitude. For instance, this can be  done by taking multiple measurements, each one with a different known mask~\cite{candes2015phase,jaganathan2015phase,iwen2016phase}. An important special case   employs  shifted versions of a single mask. The acquired data is simply  the short-time Fourier transform (STFT) magnitude of the underlying signal.
It has been shown that  in many setups, this information is sufficient for  efficient and stable recovery~\cite{bendory2018non,jaganathan2016stft,eldar2015sparse,bojarovska2016phase,griffin1984signal,pfander2016robust}. 
For a recent survey of phase retrieval from a signal processing point--of--view, see~\cite{bendory2017fourier}.

 In this work, we consider an ultra-short laser pulse characterization method,
 called Frequency-Resolved Optical Gating (FROG).  FROG is a simple, commonly-used technique for full characterization of ultra-short laser pulses which enjoys good experimental performance
 \cite{trebino2012frequency,trebino1997measuring}. In order to characterize the signal,  the FROG method measures the Fourier magnitude of the product of the signal with a translated version of itself, for several different translations. The product of the signal with itself is usually performed using a second harmonic generation
 (SHG) crystal \cite{delong1994frequency}. The acquired data, referred to as \emph{FROG trace}, 
 is a quartic function of the underlying signal and can be thought of as
 \emph{phaseless quartic  Fourier measurements}.  We refer to the problem of recovering a signal from its FROG trace as the  \emph{quartic phase retrieval problem}. Illustration of the FROG setup is given in Figure~\ref{fig:frog}.
 
\begin{figure}
	\centering
	\includegraphics[scale=0.7]{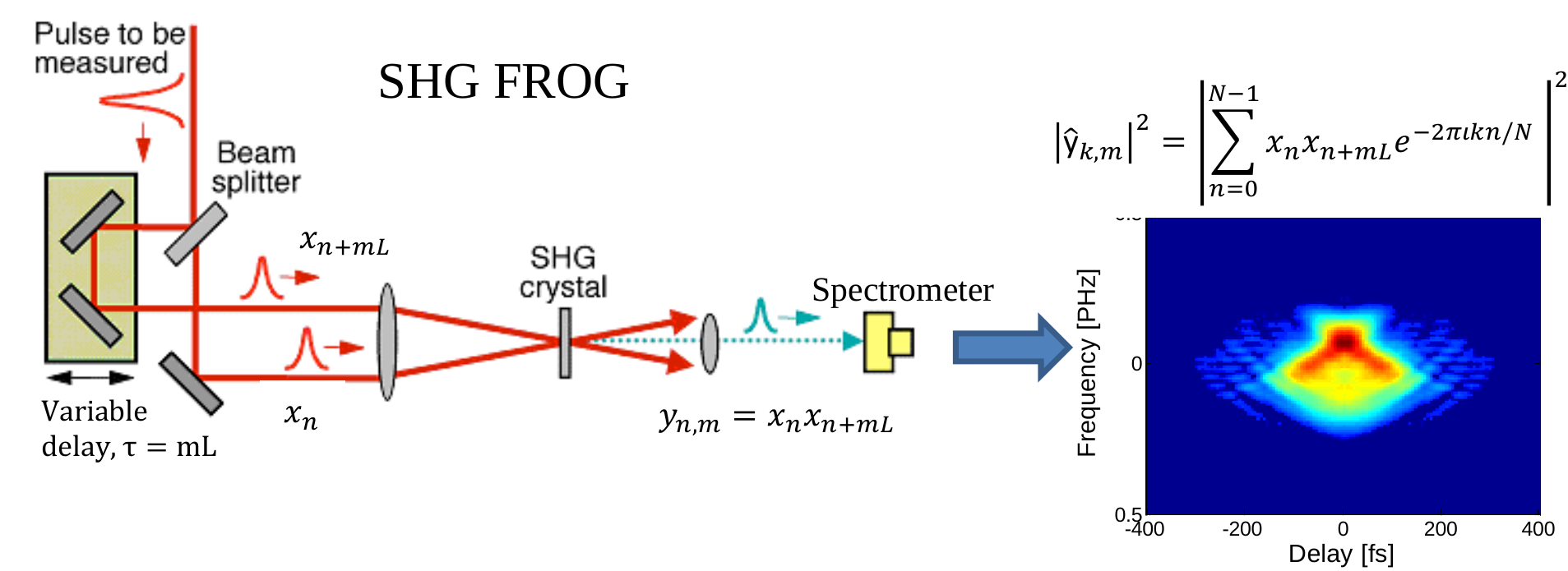}
	\label{fig:frog}
	\caption{ Illustration of the SHG FROG technique (courtesy of~\cite{bendory2017uniqueness}).}
\end{figure} 

In this paper we provide sufficient conditions on the number of samples required to determine a bandlimited signal uniquely, up to trivial ambiguities, from its FROG trace. Particularly, we show that it is sufficient to consider only three translations of the signal to determine  almost all bandlimited signals. 
 If one also measures the power spectrum of the signal, then  two translations suffice. 

The outline of this paper is as follows. In Section~\ref{sec:main_result} we formulate the FROG problem, discuss its ambiguities and present our main result. Proof of the main result is given in Section~\ref{sec:proof_main_result} and Section~\ref{sec:proof_supporting_lemmas} provides additional proofs for intermediate results. Section~\ref{sec:conclusion} concludes the paper and presents open questions.

Throughout the paper we use the following notation. We denote the Fourier transform of a signal $z\in \CN$ by $\hat{z}_k = \sum_{n=0}^{N-1}z_ne^{-2\pi\iota kn/N}$, where $\iota:=\sqrt{-1}$. We further use $\overline{z}$, $\Re\left\{z\right\}$ and $\Im\left\{z\right\}$ for its conjugate, real part and imaginary part, respectively.
We reserve $x\in\CN$ to be the underlying signal.  In the sequel, all signals are assumed to be periodic with period $N$ and all indices should be considered as modulo $N$, i.e., $z_n = z_{n+N\ell}$ for any integer $\ell\in\mathbb{Z}$.

\section{Mathematical formulation and main result} \label{sec:main_result}

The goal of this paper is to derive the minimal number of measurements required to determine a signal from its FROG trace. To this end, we  first formulate the FROG problem and  identify its symmetries, \rev{usually called \emph{trivial ambiguities} in the phase retrieval literature.} Then,  we introduce and discuss the  main results of the paper. 
 
 \subsection{The FROG trace}
   
 Let us define the signal 
\begin{equation} \label{eq:y_nm}
y_{n,m} = x_nx_{n+mL},
\end{equation}
where  $L$ is a fixed positive integer.
 The FROG trace is equal to the one-dimensional Fourier magnitude of $y_{n,m}$ for each fixed $m$, i.e., 
\begin{equation} \label{eq:frog_trace}
\begin{split}
\left\vert\hat{y}_{k,m} \right\vert^2 &=
\left\vert \sum_{n=0}^{N-1}x_nx_{n+mL}e^{-2\pi\iota nk/N }\right\vert^2, \\ \quad k=0,&\ldots,N-1, \quad m = 0,\ldots, \left\lceil N/L\right\rceil -1.
\end{split}
\end{equation}
To ease notation, we assume \rev{hereafter} that $L$ divides $N$. 

Our analysis of the FROG trace holds for bandlimited signals. Formally, we define  a bandlimited signal as follows:
\begin{defn} \label{def:bandlimited}
	We say that $x\in\mathbb{C}^N$ is a \emph{B-bandlimited signal} if its Fourier transform $\hx\in\mathbb{C}^N$ contains  $N-B $ consecutive zeros. \rev{That is, there exits $i$  such that $\hx_i=\ldots = \hx_{i+N-B-1}=0$, where all indices are taken modulo $N$.}
\end{defn}

The FROG trace is an intensity map $\CN\mapsto\mathbb{R}^{N\times \frac{N}{L}}$ that has three kinds of symmetries. These symmetries form the group of operations acting on the signal for which the intensity map is invariant. The FROG trace is invariant to global rotation, global translation and reflection~\cite{bendory2017uniqueness}. The first symmetry is continuous, while the latter two  \rev{generally} are discrete. 
 These symmetries are similar to
equivalent results in phase retrieval, see for instance~\cite{beinert2015ambiguities,bendory2017fourier}. For bandlimited signals, the global translation symmetry is also continuous:
\begin{proposition} \label{prop:ambiguities}
	Let ${x}\in\CN$ be the underlying signal and let $\hx\in\CN$ be its Fourier transform. Let $\left\vert\hat{y}_{k,m} \right\vert^2$ be the {FROG} trace of $x$ as defined in~\eqref{eq:frog_trace}  for some fixed $L$. Then, the following signals have the same FROG trace  as $x$:
	\begin{enumerate}
		\item the rotated signal $xe^{\iota\psi}$
		for some $\psi\in\mathbb{R};$
		\item the translated signal $x^{\ell}$ obeying $x^{\ell}_n = x_{n-\ell}$ for some $\ell\in\mathbb{Z}$ (equivalently, a signal  with Fourier transform $\hx^{\ell}$ obeying $\hx^{\ell}_k= \hx_k e^{-2\pi\I\ell k/N }$  for some $\ell\in\mathbb{Z}$)$;$
		\item  the reflected signal $\tilde x$ obeying $\tilde{x}_n = \overline{x_{-n}}$.
	\end{enumerate} 
 For B-bandlimited signals as in Definition~\ref{def:bandlimited} with $B\leq  N/2 $,  the translation ambiguity is continuous. Namely, any signal with a Fourier transform such that $\hx^{\psi}_k= \hx_k e^{\I \psi k }$ for some $\psi\in \mathbb{R}$ has the same FROG trace as $x$.
  
\begin{proof}
	The result for general signals was derived in~\cite{bendory2017uniqueness}. The result on the continuity of the translation symmetry for bandlimited signals is a direct corollary of Proposition~\ref{eq:pros_symmetry}, given in Section~\ref{sec:global_translation_ambiguity}.
\end{proof}
\end{proposition}

Figure~\ref{fig:symmetries} shows the  5-bandlimited  signal $x\in\mathbb{R}^{11}$ with Fourier transform given by $\hx = (1,\I,-\I,0 , 0,0,0,0,0,\I , -\I )^T$.
The second signal is $x$ shifted by three entries. 
A third signal is a  ``translated'' version of the underlying signal by 1.5 entries. Namely, the $k$th entry of its Fourier transform is $\hx_k e^{-2\pi \I(1.5) k/N}$. Clearly, the third signal is not a translated version of  $x$. Nonetheless, since $x$ is bandlimited, all three signals have the same FROG trace. If $x$ was not a bandlimited signal then the FROG trace of the \rev{latter} signal would not in general be equal to those of the first two.

\begin{figure}[h]
	\centering
	\includegraphics[scale=.8]{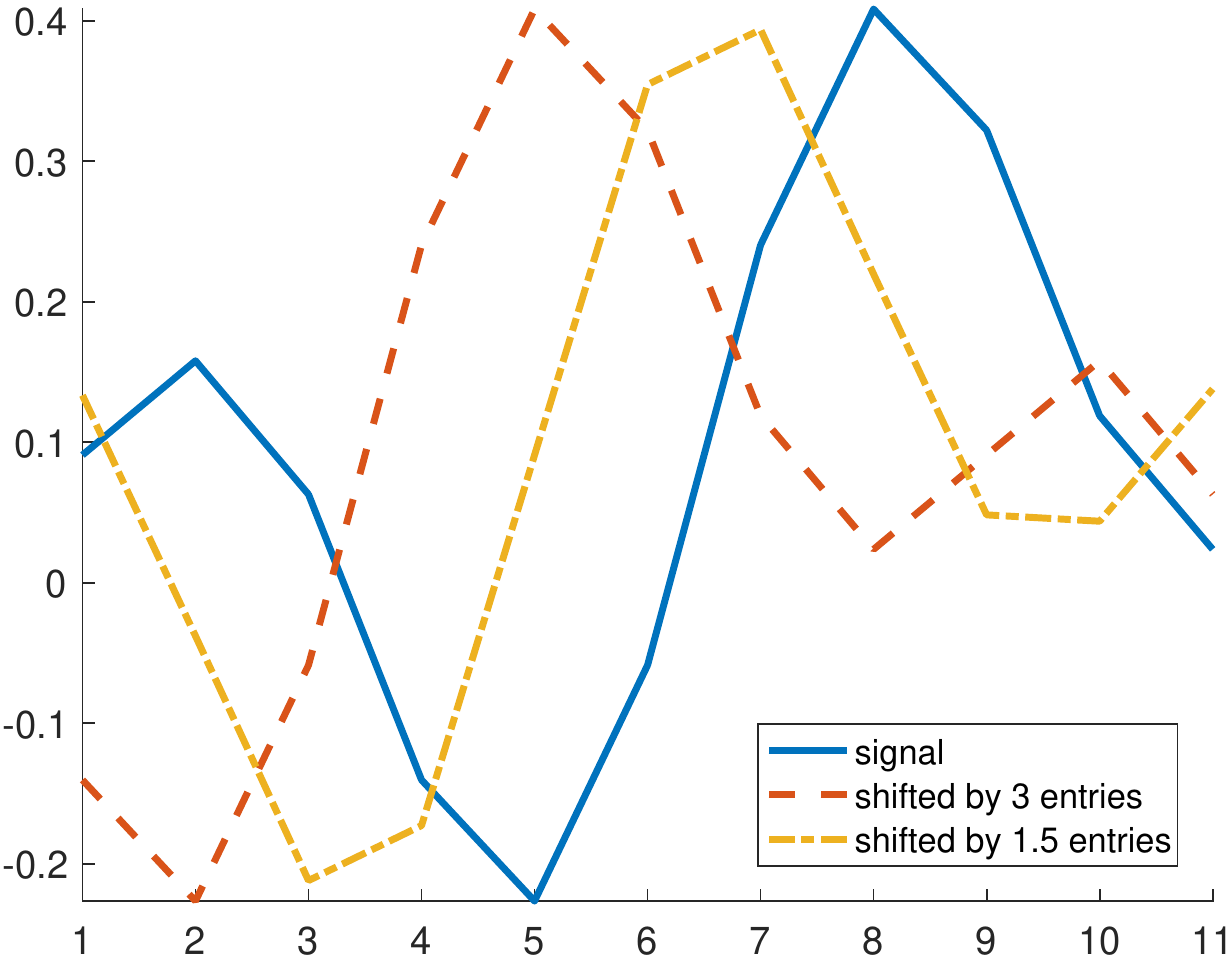}
	\caption{This figure presents three bandlimited signals in $\mathbb{R}^{11}$. Two of the them are shifted versions of each other. The third one is ``shifted'' by 1.5 entries, namely, the $k$th entry of its Fourier transform was modulated by $e^{-2\pi\I (1.5)k/N}$.  All signals have the same FROG trace.}
	\label{fig:symmetries}
\end{figure}

The symmetries of the FROG trace form a group. Namely, the \rev{FROG} intensity map $\CN\rightarrow\mathbb{R}^{N\times N/L}$ is invariant under the action of the group $G = S^1\times \mu_N \ltimes \mu_2$, where $\ltimes$ denotes a semi-direct product. The $\mu_2$ corresponds to  reflection symmetry and  $\mu_N$ corresponds to  translation ambiguity, which  rotates $\hx_k$ by  $e^{-2\pi\I\ell k /N}$  for some integer $\ell\in\mathbb{Z}$.
Observe that we use  a semi-direct product for the last symmetry since  $\mu_2$  and  $\mu_N$ do not commute; if one reflects the signal and then multiplies it by a power of $e^{2\pi\I/N}$, it is 
not the same  as multiplying by a power of $e^{2\pi\I/N}$ and then  reflecting.  In fact, the semi-direct product of $\mu_N$ and $\mu_2$ is the \emph{dihedral group} $D_{2N}$ of symmetries of the regular N-gon. 
If we consider bandlimited signals, then the FROG trace is invariant under the action of the group $G = S^1\times S^1 \ltimes \mu_2$. That is, the translation ambiguity is continuous.

\subsection{FROG recovery}

We are now ready to present the main result of this paper.
\rev{We assume throughout that the signal is bandlimited, which is  typically the case in standard ultra-short pulse characterization experiments~\cite{o2004practical}}.
We prove that almost any B-bandlimited signal is determined by its FROG trace, up to trivial ambiguities, as long as $L\leq N/4$. Particularly, we show that we need to consider  only three translations and hence    $3B$ measurements are enough to determine the underlying signal. For instance, if $L= N/4$ then the measurements corresponding to $m=0,1,2$ determine the signal.
If in addition we have access to the signal's power spectrum, then it suffices to choose $L\leq N/3$. In this case, one may consider only two translations. As an example, if $L= N/3$, then one can choose $m=0,1$ (see Remark~\ref{ren:reflect}).
The power spectrum
of the sought pulse  is often available, or can be measured by a spectrometer, which is integrated into a typical FROG device.   

\begin{thm} \label{th:main}
	Let $x$ be a $B$-bandlimited signal as defined in Definition~\ref{def:bandlimited} for some $B\leq  N/2$. If  $N/L\geq 4$, then generic signals are determined uniquely from their FROG trace as in~\eqref{eq:frog_trace}, modulo the trivial ambiguities (symmetries) of Proposition~\ref{prop:ambiguities}, from  $3B$  measurements.
	If in addition we have access to the signal's power spectrum and $N/L\geq 3$, then   $2B$ measurements are sufficient.
\end{thm}

By the notion {\em generic}, we mean that the set of signals which cannot be uniquely determined, up to trivial ambiguities, is contained in the vanishing locus of a nonzero polynomial. This implies that we can reconstruct almost all signals.

This result significantly improves upon earlier work  on the uniqueness of the FROG method. In~\cite{seifert2004nontrivial} it was shown that a continuous signal is determined by its continuous FROG trace and its power spectrum. The uniqueness of the discrete case, as the problem appears in practice, was first considered in~\cite{bendory2017uniqueness}. It was proven that a discrete bandlimited signal is determined by all $N^2$ FROG measurements (i.e., $L=1$) and the signal's power spectrum. Our result requires only $2B$ FROG measurements if the signal's power spectrum is available, where $B$ is the bandlimit. Furthermore, this is the first result showing that the FROG trace is sufficient to determine the signal \rev{even} without the power spectrum information.

It is interesting to view our results in the broader perspective of nonlinear phaseless systems of equations. In~\cite{conca2015algebraic}, it was shown that  $4N-4$ quadratic equations arising from random frame measurements are sufficient to uniquely determine \emph{all} signals. 
Another related setup is the phaseless STFT measurements. This case resembles the FROG setup, where a known reference window replaces the unknown shifted signal.
Several works derived uniqueness results for this case under different conditions~\cite{bendory2018non,eldar2015sparse,bojarovska2016phase,griffin1984signal}. In~\cite{jaganathan2016stft} it was shown that, roughly speaking, it is sufficient to set $L\approx N/2$  (namely, $2N$ measurements) to determine almost all non-vanishing signals. Comparing to Theorem~\ref{th:main}, we conclude that, maybe surprisingly, the FROG case is not significantly harder than the phaseless STFT setup.

Before moving forward to the proof of the main result, we mention that  several algorithms exist to estimate a signal from its FROG trace \cite{trebino1993using,sidorenko2016ptychographic}.  One popular
iterative algorithm is the principal components generalized
projections (PCGP) method~\cite{kane2008principal}. In each iteration, PCGP
performs principal components analysis (PCA) on a data matrix constructed by the  previous estimation. Another approach may be to minimize the least-squares loss function:
\begin{equation} \label{eq:ls}
\min_{z\in\CN} \frac{1}{2} \sum_{k=0}^{N-1} \sum_{m=0}^{N/L-1} \left( \left\vert\hat{y}_{k,m} \right\vert^2 -
\left\vert \sum_{n=0}^{N-1}z_nz_{n+mL}e^{-2\pi\iota nk/N }\right\vert^2
\right)^2.
\end{equation}
The loss function~\eqref{eq:ls} is a smooth function -- a polynomial of degree eight in $z$ -- and therefore can be minimized by standard gradient techniques. However, as the function is nonconvex, it is likely that the algorithm will  converge to a local minimum rather than to the global minimum. 

Figure~\ref{fig:tr} examines  \rev{the numerical properties of minimizing the loss function~\eqref{eq:ls} by a trust-region algorithm using the optimization toolbox Manopt~\cite{boumal2014manopt}}. The underlying signal $x\in\mathbb{R}^{24}$ was drawn from a normal i.i.d.\ distribution with mean zero and variance one. The algorithm was 
initialized from the point $x_0 = x + \sigma \zeta $, where $\sigma$ is a fixed constant and  $\zeta$ takes the values $\{-1,1\}$ with equal probability. Clearly, for $\sigma=0$, $x_0=x$ and therefore any method will succeed. We examined the empirical success for varying values of $\sigma $ and $L$. 
\rev{ As can be seen, even with $L=1$, the iterations do not always  converge from an arbitrary initialization. 
This experiment underscores the challenge in recovering the signal, even in situations where uniqueness is guaranteed.
That being said,}
the results also suggest that for low values of $L$, namely, large redundancy in the measurements, the algorithm often converges to the global minimum even when it is initialized fairly far away. In other words, the  basin of attraction of~\eqref{eq:ls} is not too small. 
Theoretical analysis of the basin of attraction in the related  problem of  random systems of quadratic equations was performed in~\cite{candes2015Wirtinger,chen2017solving,sun2016geometric,wang2017solving}. Recently, nonconvex methods for Fourier phase retrieval, accompanied with theoretical analysis, were proposed in~\cite{bendory2018non,iwen2016phase,pfander2016robust}.
However, in the FROG setup the problem is quartic rather than quadratic.

\begin{figure}[h]
	\centering
	\includegraphics[scale=.8]{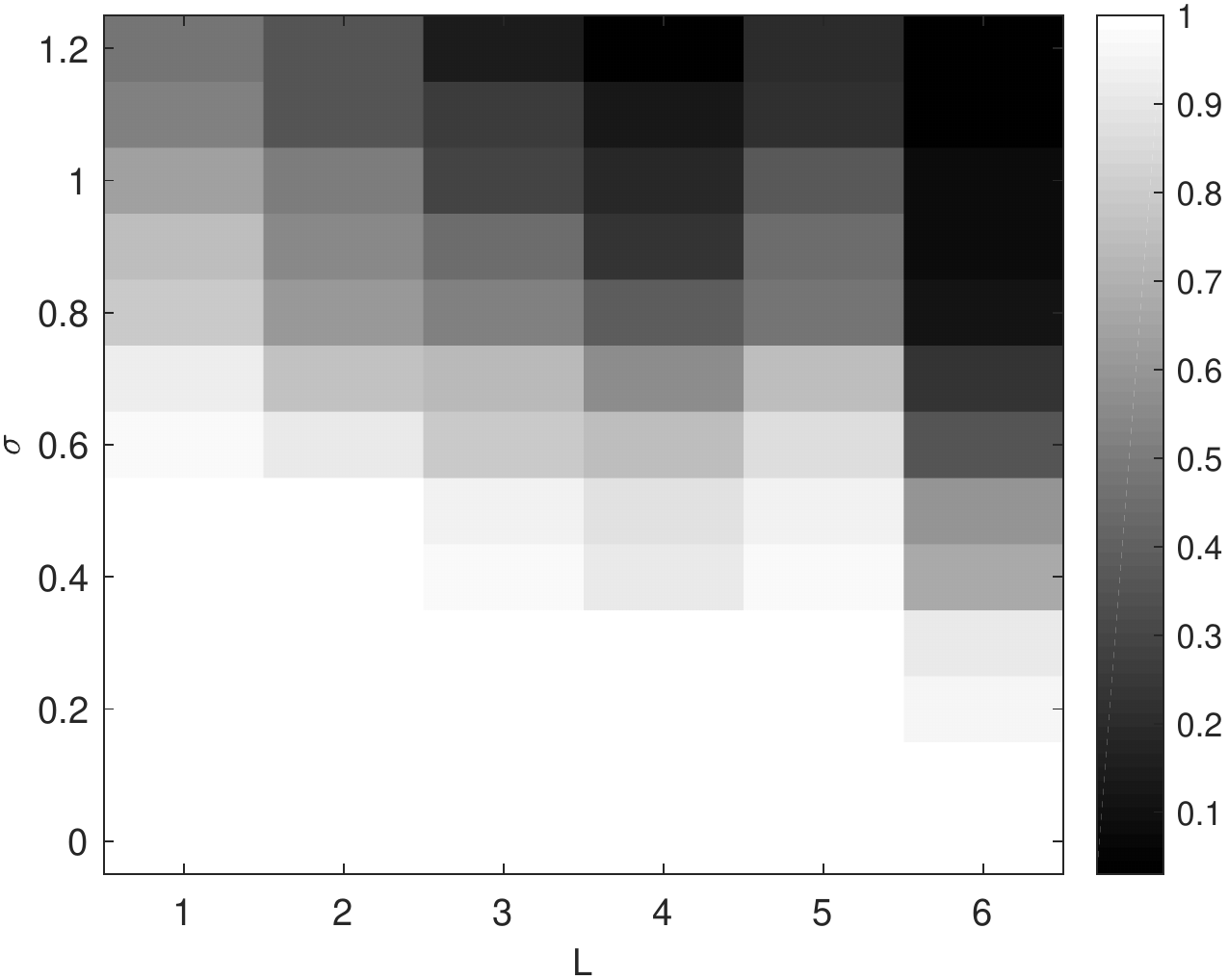}
	\caption{The empirical success recovery of a trust-region algorithm to minimize~\eqref{eq:ls} as a function of $L$ and $\sigma$ (100 experiments for each pair of values). A success was declared for recovery error less than~$10^{-6}$. 
	}
	\label{fig:tr}
\end{figure}


\section{Proof of Main Result}\label{sec:proof_main_result}

\subsection{Preliminaries}

We begin the proof by reformulating the measurement model to a more convenient structure. Applying the inverse Fourier transform we  write  $x_n=\frac{1}{N}\sum_{k=0}^{N=1}\hx_ke^{2\pi\I kn/N}$.
 Then, according to~\eqref{eq:y_nm}, we have
\begin{align*} 
\hat{y}_{k,m} &=
 \sum_{n=0}^{N-1}x_nx_{n+mL}e^{-2\pi\I kn/N } \nonumber  \\
 & = {\frac{1}{N^2}} \sum_{n=0}^{N-1}\left( \sum_{\ell_1 =0}^{N-1}\hx_{\ell_1}e^{2\pi\I\ell_1n/N} \right) \left( \sum_{\ell_2 =0}^{N-1}\hx_{\ell_2}e^{2\pi\I\ell_2n/N} e^{2\pi\I\ell_2mL/N} \right)e^{-2\pi\I k n/N} \nonumber \\
 & = {\frac{1}{N^2}} \sum_{\ell_1,\ell_2 =0}^{N-1} \hx_{\ell_1}\hx_{\ell_2}e^{2\pi\I\ell_2mL/N} \sum_{n=0}^{N-1} e^{-2\pi\I(k-\ell_1-\ell_2)n/N}.
\end{align*}
Since the later sum is equal to $N$ if  $k = \ell_1 + \ell_2$ and zero otherwise, we  get 
\begin{align} \label{eq:y_hat}
\hat{y}_{k,m} 
 = {\frac{1}{N}} \sum_{\ell=0}^{N-1} \hx_{\ell} \hx_{k-\ell}e^{2\pi\I\ell mL/N} = {\frac{1}{N}} \sum_{\ell=0}^{N-1} \hx_{\ell} \hx_{k-\ell}\omega^{\ell m},
\end{align}
where 
\begin{equation} \label{eq:omega}
\omega:=e^{2\pi\I/r}, \quad r:=N/L,
\end{equation}
and we assume that $N/L$ is an integer.
Equation~\eqref{eq:y_hat} implies that, for each fixed $k$, $\hat{y}_{k,m}$ provides $r = N/L$ samples from the (inverse) Fourier transform of $\hx_{\ell}\hx_{k-\ell}$. 

\begin{rem} \label{ren:reflect}
Note that   $\overline{\hat{y}}_{k,-m} = \sum_{\ell=0}^{N-1} \overline{\hx}_{\ell} \overline{\hx}_{k-\ell}\omega^{\ell m} $. Because of the reflection ambiguity in Proposition~\ref{prop:ambiguities}, it implies that the FROG trace is invariant to sign flip of $m$. For instance, for $r=3$, the equations for $m=1$ and $m=2$ will be the same since $m=2$ is equivalent to $m=-1$. 
\end{rem}

%

The  proof is based on recursion. We begin by showing explicitly how the first entries of $\hx$ are determined from the FROG trace. Then, we will show that the knowledge of the first $k$ entries of $\hx$ and the FROG trace is enough to determine the $(k+1)$th entry uniquely. Each recursion step is based on the results summarized in the following lemma. The lemma identifies  the number of solutions of a  system with three  phaseless equations.

\begin{lem} \label{lem:norm_linear_comb}
 	Consider the system of equations
 	\begin{equation} \label{eq:sys_eq}
 	\vert z+v_1\vert = n_1,\quad \vert z+v_2\vert = n_2,\quad \vert z+v_3\vert = n_3,
 	\end{equation}  
 	for nonnegative $n_1,n_2,n_3\in\mathbb{R}$. 
 	\begin{enumerate}
 		\item \label{lem:conda} Let $v_1,v_2,v_3\in\mathbb{C}$ be distinct and suppose that  $\Im\left\{ \frac{v_1-v_2}{v_1-v_3}  \right\}\neq 0$.  If  the system~\eqref{eq:sys_eq} has a solution, then it is unique. Moreover, if $n_1,n_2,n_3$ are fixed for generic $v_1,v_2,v_3\in \mathbb{C}$ then the system will have no solution.
 		\item \label{lem:condb} Let $v_1,v_2,v_3\in\mathbb{R}$. If $z=a+\I b$ is a solution, then $\overline{z}=a-\I b$ is a solution as well. Hence, if the system has a solution, then it has two solutions. Moreover, if $n_1,n_2,n_3$ are fixed for generic $v_1,v_2,v_3\in\mathbb{R}$ then the system will have no solution.
 		\end{enumerate}
\begin{proof}
	See Section~\ref{sec:proof_lem_norm_linear}.
\end{proof}
\end{lem}

The notion of generic signals refers here to a set of signals that are not contained in the zero set of some nonzero polynomial in the real and complex parts of each component. Consequently, since the zero set of a polynomial has strictly smaller dimension than the polynomial, this means that the set of signals failing to satisfy the conclusion of the theorem will necessarily have measure zero. 

Note that Lemma~\ref{lem:norm_linear_comb} can be extended to systems of $s\geq 3$ equations, i.e., $$\vert z+v_1\vert = n_1,\dots,\vert z+v_s\vert=n_s.$$
If one of the ratios $\frac{v_1-v_p}{v_1-v_q}$ for $p,q=2,\dots,s,\thinspace p\neq q$,  is not real then there is at most one solution to the system. 

\subsection{Proof of Theorem~\ref{th:main}}

Equipped with Lemma~\ref{lem:norm_linear_comb},  we  move forward to the proof of Theorem~\ref{th:main}.  To ease notation, \rev{we assume $B=N/2$},  $N$ is even, that $\hx_k\neq 0 $ for $k=0,\dots,N/2-1,$ and that $\hx_k =  0 $ for $k=N/2,\dots,N-1$. 
 If the signal's nonzero Fourier coefficients are not in the interval $0\dots,N/2-1$, then we can cyclically reindex the signal without affecting the proof. If $N$ is odd, then one should replace $N/2$ by $\lfloor N/2\rfloor$ everywhere in the sequel. \rev{Clearly, the proof carries through for any $B\leq N/2$.}
 
Considering~\eqref{eq:y_hat}, our bandlimit assumption \rev{on the signal} forms a ``pyramid'' structure. Here, each row represents fixed $k$ and varying $\ell$ of $\hx_\ell\hx_{k-\ell}$ \rev{for $k,\ell = 0,\ldots N/2-1$}: 
\begin{equation} \label{eq:pyramid}
\begin{array}{c} 
\hx_0^2, 0, \ldots , 0\\
\hx_0\hx_1, \hx_1 \hx_0, 0, \ldots , 0 \\
\hx_0 \hx_2, \hx_1^2, \hx_2 \hx_0, \ldots , 0
\ldots\\
\vdots\\
\hx_{{N/2}-1}\hx_0 , \hx_{{N/2}-2}\hx_1 , \ldots , \hx_{{N/2}-1}\hx_0,0, \ldots , 0\\
0, \hx_1 \hx_{{N/2}-1}, \hx_2 \hx_{{N/2}-2}, \ldots , \hx_{{N/2}-1}\hx_1,0, \ldots ,0\\
\vdots\\
0, 0, \ldots  \hx_{{N/2-1}}\hx_{{N/2-1}}, \ldots, 0, \ldots ,0.
\end{array}
\end{equation}
Then, $\hy_{k,m}$ as in~\eqref{eq:y_hat} is a subsample 
of the Fourier transform of each one of the pyramid's rows.

From the first row of~\eqref{eq:pyramid}, we see that
\begin{equation*}
\vert\hy_{0,0}\vert = \frac{1}{N} \vert\hx_0^2\vert.
\end{equation*}
Because of the \rev{continuous} rotation ambiguity, we set $\hx_0$ to be real and, without loss of generality, normalize it so that  \rev{$N\vert\hy_{0,0}\vert = \hx_0=1$}. From the second row of~\eqref{eq:pyramid}, we conclude that  
\begin{equation*}
\vert\hy_{1,0}\vert = \frac{1}{N} \vert\hx_0\hx_1+\hx_1\hx_0\vert = \frac{2}{N}\vert\hx_1\vert.
\end{equation*}
Therefore, we can determine $\vert\hx_1\vert$. Because of the \rev{continuous} translation ambiguity for bandlimited signals (see Proposition~\ref{prop:ambiguities}), we can set arbitrarily $ \hx_1 = \vert\hx_1\vert $. Note that this is not true for general signals, where the translation ambiguity is discrete.
 
Our next step is to determine  $\hx_2$ by solving the system for $k=2$. We denote the unknown variable by $z$. In this case, we obtain the system of equations for $m=0,\ldots,r-1$:
\begin{equation} \label{eq:x2}
\vert \hy_{2,m}\vert = \frac{1}{N}\left\vert \left(1+\omega^{2m}\right)z + \omega^m\hx_1^2 \right\vert,
\end{equation}
where $\omega$ is given in~\eqref{eq:omega}. Note that for $m = (2\ell+1)r/4$ for some integer $\ell\in\mathbb{Z}$, we get $\omega^{2m}=-1$ so that the system degenerates. If $r = N/L \geq 3$ then we can eliminate these equations and  the system
\begin{equation*}
\frac{\vert \hy_{2,m}\vert}{\left\vert 1+\omega^{2m}\right\vert} = \frac{1}{N}\left\vert z + \frac{\omega^m}{1+\omega^{2m}}\hx_1^2 \right\vert,
\end{equation*}
  still has at least two distinct equations.
It is easy to see that since $\vert \omega \vert  =1$,   $\frac{\omega^m}{1+\omega^{2m}}=\frac{\omega^{-m}}{1+\omega^{-2m}}$ so that this term is  self-conjugate and hence real. Since $z = \hx_2$ is a solution, by the second part of Lemma~\ref{lem:norm_linear_comb}, we conclude that the system has two conjugate solutions $z$  and $\overline{z}$, \rev{corresponding} to the reflection symmetry of Proposition~\ref{prop:ambiguities}. Hence, we fix $\hx_2$ to be one of these two conjugate solutions. 

Fixing $\hx_0,\hx_1,\hx_2$ up to symmetries, we move forward to determine $\hx_3$.  For $k=3$, we get  the system of equations for  $m=0,\dots,r-1,$
\begin{equation*}
\vert \hy_{3,m}\vert = \frac{1}{N}\left\vert z+ \omega^m\hx_1\hx_2 + \omega^{2m}\hx_2\hx_1 +\omega^{3m}z \right\vert.
\end{equation*}
As in the previous case,  for $m = (2\ell+1)r/6$ for some integer $\ell\in\mathbb{Z}$  we have $\omega^{3m} = -1$. 
In the rest of the cases, we reformulate the equations as
\begin{equation} \label{eq:x2}
\frac{\vert \hy_{3,m}\vert}{\vert 1+\omega^{3m} \vert } = \frac{1}{N}\left\vert z+ \frac{(\omega^m+\omega^{2m})}{1+\omega^{3m}} \hx_1\hx_2 \right\vert.
\end{equation}
Again, since $\vert \omega\vert=1 $, $\frac{\omega^m+\omega^{2m}}{1+\omega^{3m}}$ is self conjugate and hence real.
Let us denote $\hx_2 = \vert\hx_2\vert e^{\I\theta}$ and divide by $e^{\I\theta}$  to  obtain  
\begin{equation} \label{eq:x2_nor}
\frac{\vert \hy_{3,m}\vert}{\vert 1+\omega^{3m} \vert } = \frac{1}{N}\left\vert ze^{-\I\theta}+ \frac{(\omega^m+\omega^{2m})}{1+\omega^{3m}} \hx_1\vert \hx_2\vert  \right\vert.
\end{equation}
 Since we set $\hx_1$ to be real, this is a system of the form of the second part of Lemma~\ref{lem:norm_linear_comb}, having two conjugate solutions. Denote these solutions  by $z_1,z_2$ and recall that the candidate solutions for~\eqref{eq:x2} are $z_1e^{\I\theta}$ and $z_2e^{\I\theta}$. Since  $ \hx_3$ is a solution to~\eqref{eq:x2}, $z_1 = \hx_3e^{-\I\theta}$ is one solution. The second solution is given by $z_2 = \overline{z_1}=\overline{\hx_3}e^{\I\theta}$. Therefore, we conclude that $\overline{\hx_3}e^{2\I\theta}$ is a second potential solution to~\eqref{eq:x2}.
  So, currently we have two candidates for $\hx_3$. Next, we will determine $\hx_4$ uniquely and show that  $\overline{\hx}_3e^{2\I\theta}$  is inconsistent with the data. This will determine $\hx_3$ uniquely.

For $\hx_4$ and by eliminating the case of $m = r(2\ell+1)/8$ for an integer $\ell\in\mathbb{Z}$ (namely, $\omega^{4m} = -1$), we get the system for $m = 0,\dots,r-1,$
\begin{equation} \label{eq:x4}
\frac{\vert \hy_{4,m}\vert}{\vert 1+\omega^{4m} \vert } = \frac{1}{N}\left\vert z+ \frac{(\omega^m+\omega^{3m})\hx_1\hx_3}{1+\omega^{4m}} +\frac{\omega^{2m}\hx_2^2}{1+\omega^{4m}} \right\vert.
\end{equation}
To invoke Lemma~\ref{lem:norm_linear_comb}, we need three equations. In general, it is most convenient to choose $m=0,1,2$. If one of these values satisfy $\omega^{4m} = -1$, then we may always choose another value. Note that for $r=3,4$, which are of particular interest, $\omega^{4m}\neq -1$.  The following lemma paves the way to determining $\hx_4$ uniquely:
\begin{lem} \label{lemma.ratios}
	Let $v_m = \frac{(\omega^m+\omega^{3m})\hx_1\hx_3}{1+\omega^{4m}} +\frac{\omega^{2m}\hx_2^2}{1+\omega^{4m}} $. 
	If $r = N/L\geq 4$, then for generic $\hx_1,\hx_2,\hx_3$ the ratio $\frac{v_0-v_q}{v_0-v_p}$ is not real for some distinct $p,q \in \{1,\dots,r-1\}$ with $p +q \neq r$.
\end{lem}
\begin{proof}
	See Section~\ref{sec:lemma.ratios}.
\end{proof}
\noindent Thus, by Lemma~\ref{lem:norm_linear_comb} we conclude that if a solution for~\eqref{eq:x4} exists, then it is unique. When $r = 3$, the system above provides only two distinct equations since the FROG trace of  $m=1$ and $m=2$ is the same; see Remark~\ref{ren:reflect}. 
 However, if $\vert \hx_4\vert $ is known, then we get a third equation $\vert z\vert = \vert \hx_4\vert $ and 
a similar application of Lemma~\ref{lem:norm_linear_comb} shows that $\hx_4$ is uniquely determined. 

Recall that currently we have two potential candidates for $\hx_3$. However, we show in Section~\ref{sec:supporting_lemma_x3} that if we replace $\hx_3$ by $\overline{\hx_3}e^{2\I\theta}$ where $\theta = \arg\left(\hx_2\right) $ then for generic signals the system of equations~\eqref{eq:x4} has no solution. Therefore, we conclude that we can fix $\hx_0,\hx_1,\hx_2,\hx_3,\hx_4$ up to trivial ambiguities. 

The final step of the proof is to show that given $\hx_0,\hx_1,\dots,\hx_k$ for some $k\geq 4$, we can determine $\hx_{k+1}$ up to symmetries. For an even $k=2s$ we get the system of equations for $m=0,\dots,r-1,$
\begin{equation*}
 \left\vert\frac{\hy_{k+1,m}}{1+\omega^{m(k+1)}} \right\vert  = \frac{1}{N} \left\vert  z + \frac{\omega^m}{1+\omega^{m(k+1)}}\hx_1\hx_k+\dots+\frac{\omega^{ms}}{1+\omega^{m(k+1)}}\hx_s\hx_{s+1}   \right\vert,
\end{equation*}
where again we omit the case $\omega^{m(k+1)}=-1$. To invoke Lemma~\ref{lem:norm_linear_comb}, we need three equations. In most cases, one can simply choose $m=0,1,2$. If one of these values violate the condition $\omega^{m(k+1)}\neq -1$, then we replace it with larger values of $m$. For $k=2s+1$ odd, we obtain the system
\begin{equation*}
\left\vert\frac{\hy_{k+1,m}}{1+\omega^{m(k+1)}} \right\vert  = \frac{1}{N} \left\vert  z + \frac{\omega^m}{1+\omega^{m(k+1)}}\hx_1\hx_k+\dots+\frac{\omega^{m(s+1)}}{1+\omega^{m(k+1)}}\hx_{s+1}^2   \right\vert.
\end{equation*}
Let us assume that $k=2s$ is even and denote $v_m = \frac{\omega^m}{1+\omega^{m(k+1)}}\hx_1\hx_k+\dots+\frac{\omega^{ms}}{1+\omega^{m(k+1)}}\hx_s\hx_{s+1}$.
If $r=N/L\geq 4$, then the same argument used in the proof of Lemma~\ref{lemma.ratios}  shows that for generic values of $\hx_0, \hx_1, \ldots , \hx_k,$ the ratio
${v_0 -v_p}\over{v_0 -v_q}$
will not be real for some distinct values of $p,q$ with $p + q \neq r$. Therefore,
  the system has a unique solution by Lemma~\ref{lem:norm_linear_comb}. A similar statement holds for $k=2s+1$ odd.

When $r=3$, the system provides only two distinct equations; see Remark~\ref{ren:reflect}. If in addition we assume that knowledge of $\vert \hx\vert $, then we have an additional equation $\vert z\vert = \vert \hx_{k+1}\vert $, ensuring a unique solution.

\subsection{Example: Determining $\hx_4$ given  $\hx_0, \hx_1, \hx_2, \hx_3$}

To illustrate the method, we describe in more detail the terms used to determine $\hx_4 = a_4+\I b_4$
from $\hx_0=1, \hx_1, \hx_2, \hx_3$.
We are given the following information in a ``pyramid form'' from which we must determine the unknowns
$a_4, b_4$:

$$\begin{array}{c} | (a_4 + b_4 \I) + (a_1 + b_1\I)(a_3+b_3\I) + (a_2
+ b_2 \I)^2 + (a_3 + b_3\I) (a_1 + b_1 \I) + (a_4 +b_4\I)|\\
| (a_4 + b_4 \I) + \omega (a_1 + b_1\I)(a_3+b_3\I) + \omega^2(a_2 + b_2 \I)^2
+ \omega^3 (a_3 + b_3\I) (a_1 + b_1 \I) + \omega^4(a_4 +b_4\I)|\\
\ldots\\
| (a_4 + b_4 \I) + \omega^{r-1} (a_1 +
b_1\I)(a_3+b_3\I) + \omega^{2(r-1)}(a_2 + b_2 \I)^2 +
\omega^{3(r-1)} (a_3 + b_3\I) (a_1 + b_1 \I) + \omega^{4r-4}(a_4
+b_4\I)|.
\end{array}
$$
As we have done throughout the proof, we rearrange the terms as
$$\begin{array}{c}
| (a_4 + b_4 \I) + (a_1 + b_1\I)(a_3 + b_3\I) +  \frac{1}{2}(a_2+b_2\I)^2|\\
| (a_4 + b_4\I) + {(\omega + \omega^3)\over{1 + \omega^4}} (a_1 + b_1\I)(a_3 + b_3\I) + {\omega^2\over{1 + \omega^4}} (a_2 + b_2\I)^2|\\
\ldots\\
| (a_4 + b_4\I) + {(\omega^{r-1} + \omega^{3r-3)})\over{1 + \omega^{4r-4}}} (a_1 + b_1\I)(a_3 + b_3\I) + {\omega^{2r-2}\over{1 + \omega^{4r-4}}} (a_2 + b_2\I)^2|,
\end{array}
$$
for any $m$ satisfying $\omega^{4m}\neq -1$.
Suppose that $r = N/L \geq 4$ and $\omega^4,\omega^8\neq -1$ so we can choose the terms associated with $m=0,1,2$. If the ratio
$$\left(\hx_1\hx_3 + \frac{1}{2} \hx_2^2- {(\omega + \omega^3)\over{1 + \omega^4}} \hx_1 \hx_3 - {\omega^2\over{1 + \omega^4}} \hx_2^2\right)/\left(
\hx_1\hx_3 + \frac{1}{2} \hx_2^2- {(\omega^2 + \omega^6)\over{1 + \omega^8}} \hx_1 \hx_3 - {\omega^4\over{1 + \omega^8}} \hx_2^2 \right)$$
is not real, 
then $\hx_4 = a_4 +  \I b_4$ is uniquely determined by the 3 real numbers
$$\begin{array}{c}| (a_4 + b_4 \I) + (a_1 + b_1\I)(a_3 + b_3\I) +  \frac{1}{2}(a_2+b_2\I)^2|\\
| (a_4 + b_4\I) + {(\omega + \omega^3)\over{1 + \omega^4}} (a_1 + b_1\I)(a_3 + b_3\I) + {\omega^2\over{1 + \omega^4}} (a_2 + b_2\I)^2|\\
| (a_4 + b_4\I) + {(\omega^2 + \omega^6)\over{1 + \omega^8}} (a_1 + b_1\I)(a_3 + b_3\I) + {\omega^4\over{1 + \omega^8}} (a_2 + b_2\I)^2|. \end{array}$$
If $\omega^4 = -1$ then $\omega = e^{2\pi \I/8}$ and one can determine $a_4, b_4$ from the terms corresponding to $m=0,2,3$:
$$\begin{array}{c}| (a_4 + b_4 \I) + (a_1 + b_1\I)(a_3 + b_3\I) +  \frac{1}{2}(a_2+b_2\I)^2|\\
| (a_4 + b_4\I) + {(\omega^2 + \omega^6)\over{1 + \omega^8}} (a_1 + b_1\I)(a_3 + b_3\I) + {\omega^4\over{1 + \omega^8}} (a_2 + b_2\I)^2|\\
| (a_4 + b_4\I) + {(\omega^4 + \omega^{12})\over{1 + \omega^{16}}} (a_1 + b_1\I)(a_3 + b_3\I) + {\omega^8\over{1 + \omega^{16}}} (a_2 + b_2\I)^2|.
\end{array}$$
By directly substituting $\omega$, these terms reduce to:
$$\begin{array}{c}| (a_4 + b_4 \I) + (a_1 + b_1\I)(a_3 + b_3\I) +  \frac{1}{2}(a_2+b_2\I)^2|\\
| (a_4 + b_4\I)  -\frac{1}{2} (a_2 + b_2\I)^2|\\
| (a_4 + b_4\I) - (a_1 + b_1\I)(a_3 + b_3\I) + \frac{1}{2} (a_2 + b_2\I)^2|.
\end{array}$$
Likewise if $\omega^8 = -1$ then $\omega = e^{2\pi \I/16}$ and we determine
$a_4 + b_4 \I$ from the three real numbers corresponding to $m=0,1,4$:
$$\begin{array}{c}| (a_4 + b_4 \I) + (a_1 + b_1\I)(a_3 + b_3\I) +  \frac{1}{2}(a_2+b_2\I)^2|\\
| (a_4 + b_4\I) + {(\omega + \omega^3)\over{1 + \omega^4}} (a_1 + b_1\I)(a_3 + b_3\I) + {\omega^2\over{1 + \omega^4}} (a_2 + b_2\I)^2|\\
| (a_4 + b_4\I) + {(\omega^4 + \omega^{12})\over{1 + \omega^{16}}} (a_1 + b_1\I)(a_3 + b_3\I) + {\omega^8\over{1 + \omega^{16}}} (a_2 + b_2\I)^2|.
\end{array}$$

When $r = 3$, $\omega^3 =1$ and we only obtain two distinct numbers for $m=0,1$:
$$ \begin{array}{c} | (a_4 + b_4 \I) + (a_1 + b_1\I)(a_3 + b_3\I) +  \frac{1}{2}(a_2+b_2\I)^2|\\
| (a_4 + b_4\I) +  (a_1 + b_1\I)(a_3 + b_3\I) + {\omega^2\over{1 + \omega}} (a_2 + b_2\I)^2|. \end{array}$$
In
this case, there are two possible solutions. However, if we also
assume that we know $|\hx|$ then we have a third piece of information to
uniquely determine $\hx_4$.


\section{Proofs of supporting results} \label{sec:proof_supporting_lemmas}

\subsection{On the  translation symmetry for bandlimited signals} \label{sec:global_translation_ambiguity}

The following proposition shows that if the signal is bandlimited, then the translation symmetry is continuous.

\begin{proposition} \label{eq:pros_symmetry}
	Suppose that $x$  is a B-bandlimited signal with $B\leq N/2$. Assume without loss of generality that  $\hx_{B}=\ldots=\hx_{N-1}=0$. Then, for any $\mu=e^{\I\psi}$ for some $\psi\in[0,2\pi)$, any signal with Fourier transform
	\begin{equation*}
	[\hat x_0,  \mu\hat x_1, \mu^2\hat x_2,\dots,\mu^{B-1}\hat x_{B-1},0,\dots,0  ],
	\end{equation*}
	has the same FROG trace~\eqref{eq:frog_trace} as $x$.	
	
	\begin{proof}
		Under the bandlimit assumption, \rev{we can substitute $p=\ell$ and $q=k-\ell$ and} write~\eqref{eq:y_hat}   as 
		\begin{equation*}
		\hat{y}_{k,m} = {\frac{1}{N}} \sum_{\substack{p+q=k	 \\ 0\leq p,q\leq  N/2-1} } \hx_{p}\hx_{q}e^{2\pi\I p mL/N}.
		\end{equation*}
		Now, if  $\hat{x}_p$ is replaced by $\mu^p\hat{x}_p$ and $\hat{x}_q$ is replaced by  $\mu^q\hat{x}_q$ then $\hat{y}_{k,m}$ is replaced by $\mu^k\hat{y}_{k,m}$. Hence, the absolute value of $\hat{y}_{k,m}$ remains unchanged. \rev{Without the bandlimit assumption,  $q=(k-\ell)\mod N$ and thus $\vert \hat{y}_{k,m}\vert$ is changed unless $\mu$ is the $N$th root of unity.}
	\end{proof}
\end{proposition}
%

\subsection{Proof of Lemma~\ref{lem:norm_linear_comb}} \label{sec:proof_lem_norm_linear}
The system of equations~\eqref{eq:sys_eq} can be written explicitly as 
	\begin{align} \label{eq:1}
	\vert z\vert^2 + \vert v_1\vert^2+2\Re\left\{ z\overline{v_1} \right\} &= n_1^2,\nonumber \\
	\vert z\vert^2 + \vert v_2\vert^2+2\Re\left\{ z\overline{v_2} \right\} &= n_2^2, \\
	\vert z\vert^2 + \vert v_3\vert^2+2\Re\left\{ z\overline{v_3} \right\} &= n_3^2. \nonumber
	\end{align}
	Subtracting the  the second and the third equations from the first, we get 
	\begin{align*}
	\Re\left\{ z\left( \overline{v}_1 - \overline{v}_2 \right) \right\} &= \frac{1}{2}\left( n_1^2-n_2^2 + \vert v_2\vert^2 - \vert v_1\vert^2\right), \nonumber\\ 
	\Re\left\{ z\left( \overline{v}_1 - \overline{v}_3 \right) \right\} &= \frac{1}{2}\left( n_1^2-n_3^2 + \vert v_3\vert^2 - \vert v_1\vert^2\right). 
	\end{align*} 
	Let $z = a + \I b$, $\overline{v}_1- \overline{v}_2 = c+ d\I$ and $\overline{v}_1- \overline{v}_3 = e+ f\I$. Then, we obtain a system of two linear equations for two variables:
	\begin{align} \label{eq:2}
	ac - bd &= \frac{1}{2}\left( n_1^2-n_2^2 + \vert v_2\vert^2 - \vert v_1\vert^2\right), \nonumber\\ 
	ae - bf &= \frac{1}{2}\left( n_1^2-n_3^2 + \vert v_3\vert^2 - \vert v_1\vert^2\right). 
	\end{align} 
	This system has a unique solution provided that the vectors $(c,-d)$ and $(e,-f)$ are not proportional. This is equivalent to the assumption   $\Im\left\{ \frac{v_1-v_2}{v_1-v_3}  \right\}\neq 0$.

We now show that for generic $v_1,v_2,v_3$ the system~\eqref{eq:1} has no solutions, so the unique solution to the linear system~\eqref{eq:2} will not be a solution to~\eqref{eq:1}.  If we express $v_k = a_k+\I b_k$, we consider the variety $\mathcal{I}\subset(\mathbb{R}^2)^4$ of tuples 
\begin{equation} \label{eq:tuples}
((a,b),(a_1,b_1),(a_2,b_2),(a_3,b_3)),
\end{equation}
such that $z=a +\I b$ is a solution to~\eqref{eq:1}. 	
Each of the  equations in~\eqref{eq:1} involves a different set of variables, so it imposes an independent condition on the tuples~\eqref{eq:tuples}. It follows that dim$\mathcal{I}\leq 8-3=5$.
In particular, $\mathcal{I}$ has strictly smaller dimension than the $\mathbb{R}^6$ parametrizing all triples
$((a_1,b_1),(a_2,b_2),(a_3,b_3))$. Therefore, the system~\eqref{eq:1} has no solution for generic $v_k$.
Intuitively, this can be seen by noting that the set of $z$ satisfying the equation $\vert z+v_k\vert = n_k$
is a circle of radius $n_k$ centered at $-a_k-\I b_k$ in the complex plane. For generic choices of
centers, three circles of fixed radii $n_1,n_2,n_3$ will have no intersection.
	
	
	For the second part of the lemma, suppose that $v_1,v_2,v_3\in \mathbb{R}$. In this case, \eqref{eq:sys_eq} is simplified to
	\begin{align*}
	\left(a + v_1\right)^2 + b^2 &= n_1^2,  \nonumber \\
	\left(a + v_2\right)^2 + b^2 &= n_2^2, \\
	\left(a + v_3\right)^2 + b^2 &= n_3^2. \nonumber
	\end{align*}
	These equations are invariant under the transformation $(a,b)\mapsto(a,-b)$ so if $z=a+b\I$ is a solution then $\overline{z}=a - b\I$ is a solution as well. Subtracting any pair of equations gives a linear equation in $a$, so there is at most one value of $a$ solving the system. Hence, if the system has a solution, then it has two conjugate solutions. 
	
	\subsection{Proof of Lemma~\ref{lemma.ratios}} \label{sec:lemma.ratios}
	
		Let $\alpha_{m,1} = {\omega^{3m} + \omega^m\over{1 + \omega^{4m}}}$ and
		$\alpha_{m,2} = { \omega^{2m}\over{1 + \omega^{4m}}}$.  As noted above,
		$\alpha_{m,1}$ and $ \alpha_{m,2}$ are real if they are well defined. Additionally, note that $\alpha_{m,i} = \alpha_{r-m,i}$ for $i=1,2$.
		 If $r = N/L
		\geq 4$ then we can find $p,q$ with $p + q \neq r$
		so that $\alpha_{p,1}, \alpha_{p,2}, \alpha_{q,1}, \alpha_{q,2}$ are well defined. For instance, if $\omega^4 \neq -1, \omega^8 \neq -1$  we take $p=1, q=2$.
		Then we can write
		\begin{equation} \label{eq.ratio}
		\frac{v_0 -v_p}{v_0 -v_q} = \frac{( \alpha_{0,1} - \alpha_{p,1})\hx_1 \hx_3 +
			(\alpha_{0,2} - \alpha_{p,2}) \hx_2^2}
		{( \alpha_{0,1} - \alpha_{q,1})\hx_1 \hx_3 + (\alpha_{0,2} - \alpha_{q,2}) \hx_2^2}.
		\end{equation}
		Moreover, the  $\alpha_{m,i}$ are fixed, so this ratio is well defined as long
		as the $\hx_1, \hx_2,\hx_3$ are not solutions to the nonzero quadratic polynomial
		$ ( \alpha_{0,1} - \alpha_{q,1})\hx_1 \hx_3 + (\alpha_{0,2} - \alpha_{q,2}) \hx_2^2$. Now we can multiply the numerator and denominator of~\eqref{eq.ratio} 
		by $\overline{(v_0 - v_q)}$
		to obtain
		$${1\over{|v_0 - v_q|^2}} \left(\beta_{p,1}\beta_{q,1}|\hx_1\hx_3|^2 +
		\beta_{p,2}\beta_{q,2}|\hx_2|^2 +
		\beta_{p,1}\beta_{q,2}\hx_1\hx_3\overline{\hx_2^2} +
		\beta_{p,2}\beta_{q,1}\overline{\hx_1\hx_3}
		\hx_2^2\right),$$
		where $\beta_{s,i} = \alpha_{0,i} - \alpha_{s,i}$ for $s = p,q$ are fixed nonzero real numbers.
		This expression is real only if
		$$\Im\left((\alpha_{0,1}- \alpha_{p,1})(\alpha_{0,2} - \alpha_{q,2})\hx_1\hx_3\overline{\hx_2^2} +
		(\alpha_{0,2}- \alpha_{p,2})(\alpha_{0,1} - \alpha_{q,1})\overline{\hx_1\hx_3}
		\hx_2^2\right)=0.$$
		This condition is a quartic polynomial in the real and complex components of $\hx_1, \hx_2, \hx_3$. Hence for generic signals the left hand side will not be equal to zero.

\subsection{Determining  $\hx_3$ uniquely} \label{sec:supporting_lemma_x3}	

The following lemma shows that given $\hx_1,\hx_2$ and $\hx_4$, $\hx_3$ is determined uniquely from the FROG trace up to symmetries.

\begin{lem}
	For a generic signal, let $\hx_1\in\mathbb{R}$, $\hx_2,\hx_3,\hx_4\in\mathbb{C}$ with $r\geq 4$. Consider the following system of equations $m = 0,\dots,r-1$:
	\begin{equation} \label{eq:3}
	\left\vert z + \frac{\left(\omega^m+\omega^{3m}\right)}{1+\omega^{4m}}\hx_1\hx_3^{'} + \frac{\omega^{2m}}{1+\omega^{4m}}\hx_2^2\right\vert = \left\vert \hx_4 + \frac{\left(\omega^m+\omega^{3m}\right)}{1+\omega^{4m}}\hx_1\hx_3 + \frac{\omega^{2m}}{1+\omega^{4m}}\hx_2^2\right\vert.
	\end{equation}
If $\hx_3^{'}=e^{2\theta\I}\overline{\hx_3}$ and $\theta = \arg\left(\hx_2\right)$, then the system has no solutions. Moreover, if $r=3$, then the system of equations:
\begin{eqnarray*}
\vert z \vert & = & \vert \hx_4\vert  \\
\vert z + \hx_1\hx_3^{'} + \hx_2^2/2 \vert & = & \vert \hx_4+\hx_1\hx_3+  \hx_2^2/2\vert  \\
\left\vert z + \frac{\omega+1}{1+\omega}\hx_1\hx_3^{'} + \frac{\omega^2}{1+\omega}\hx_2^2 \right\vert & = & \left\vert \hx_4 + \frac{\omega+1}{1+\omega}\hx_1\hx_3 + \frac{\omega^2}{1+\omega}\hx_2^2 \right\vert
\end{eqnarray*}
has no solutions.
\end{lem}
\begin{proof}
	
	The proof is similar to the proof of the first part of Lemma~\ref{lem:norm_linear_comb}. Since $\hx_1$ is generic, we can assume it is nonzero. Let us denote $\hx_k = a_k + \I b_k$ for $k=2,3,4.$ If the system of equations has a solution, then $a_2,b_2,a_3,b_3,a_4,b_4$ satisfy a non-trivial polynomial equation. To see this, we argue the following:
	
	By assuming $\hx_1\neq 0$, we can divide the equations  by $\hx_1$ and then reduce to the case $\hx_1 = 1$. The proof of Lemma~\ref{lem:norm_linear_comb} shows that if a solution $z= a+\I b$ exists, then by considering the differences of equations for three distinct values of $m$, $a$ and $b$ can be expressed as rational functions in the real and complex parts of $\hx_2,\hx_3,\hx_3^{'}$.  Since $\hx_3^{'} = \left(\frac{\hx_2^2}{\vert \hx_2\vert} \right)^2\overline{\hx_3}$, $\Re\left(\hx_3^{'}\right)$ and $\Im\left(\hx_3^{'}\right)$ are also rational functions of $a_2,b_2,a_3,b_3$. Hence, if $z = a+\I b$ is a solution to the system~\eqref{eq:2}, then $a = f(a_2,b_2,a_3,b_3,a_4,b_4)$ and $b = g(a_2,b_2,a_3,b_3,a_4,b_4)$, where $f$ and $g$ are rational functions. 
	 
	 In order for $z=a+\I b$ to be a feasible solution, it must also satisfy the quadratic equation for $m=0$
	 \begin{equation*}
	 \left\vert z + \hx_3^{'}\hx_1 + \hx_2^2/2 \right\vert^2   = \left\vert \hx_4 + \hx_3\hx_1 + \hx_2^2/2 \right\vert^2.
	 \end{equation*}
	Expanding both sides in terms of $(a_2,b_2,a_3,b_3,a_4,b_4)$, we see that the real numbers $(a_2,b_2,a_3,b_3,a_4,b_4)$ must satisfy an explicit polynomial equation $F(a_2,b_2,a_3,b_3,a_4,b_4)=0$. As long as $F\neq 0$, the generic real numbers will not satisfy it. 
	
We are left with showing that indeed $F\neq 0$. 
For example if $\omega^4, \omega^8 \neq -1$  then we can use the equations associated with $m=0,1,2$:
\begin{equation*}
\begin{array}{ccc}
|x + \hx_1 \hx_3' + \hx_2^2/2|^2 & = & |\hx_4 + \hx_1 \hx_3 + \hx_2^2/2|^2 \\
|x + {\omega + \omega^3\over{1 +\omega^4}}   \hx_1 \hx_3'  +
{\omega^2\over{1 + \omega^4}} \hx_2^2|^2 & = &
|\hx_4 + {\omega + \omega^3\over{1 +\omega^4}} \hx_1 \hx_3  +  {\omega^2\over{1 + \omega^4}} \hx_2^2|^2\\
|x + {\omega^2 + \omega^6\over{1 +\omega^8}}  \hx_1 \hx_3'  +
{\omega^4\over{1 + \omega^8}} \hx_2^2|^2 & = &
|\hx_4 + {\omega^2 + \omega^6\over{1 +\omega^8}} \hx_1 \hx_3  +  {\omega^2\over{1 + \omega^4}} \hx_2^2|^2,
\end{array} 
\end{equation*}
to show that $a_2, b_2, a_3, b_3, a_4, b_4$
satisfy a nonzero polynomial $F(a_2, b_2, a_3,b_3, a_4, b_4)$. This is true since 
$F(0,1, 1,1,a_4,b_4)$  (corresponding to $\hx_1=1, \hx_2 = \I, \hx_3 = 1+ \I$)
has coefficient  of $a_4^2$ given by 
$$4096 \csc(8 t)^2 \sin(t/2)^8 \sin(t)^2 \sin(2 t)^2 (\sin(t) + \sin(2 t))^2,$$
where $\omega = \cos t + \I \sin t$ and $\csc$ is the inverse of the sine function. This coefficient is nonzero unless
$t=0$ or $t = 2\pi/3$. 
A similar analysis can be performed in the other cases. For instance, instead of picking $m=0,1,2,$ if $\omega^4 = -1$ one may take $m=0,2,4$ and if $\omega^8 =-1$ then we choose $m=0,1,3$. Hence, we conclude that indeed $F\neq 0$.

If $r=3$, then~\eqref{eq:3} provides only two distinct equations (see Remark~\ref{ren:reflect}). We then use the third constraint $\vert z \vert =\vert\hx_4 \vert $ to derive the same result.

\end{proof}


\section{Conclusion and perspective} \label{sec:conclusion}

FROG is an important tool for  ultra-short laser pulse characterization.
The problem involves a system of phaseless quartic equations that  differs significantly from quadratic systems, appearing in \rev{standard} phase retrieval \rev{problems}.  
In this work, we  analyzed the uniqueness of the FROG method.
We have shown that it is sufficient to take only  $3$B FROG measurements in order to  determine a generic B-bandlimited signal uniquely, up to unavoidable symmetries. If the power spectrum of the sought signal is also available, then $2B$ \rev{FROG} measurements are enough. \rev{Necessary conditions for recovery is an open problem. In addition, this paper did not study   computational aspects.  An important step in this direction is to analyze the properties of current algorithms used by practitioners, like the PCGP.} 

\rev{A natural extension of the FROG model  is called} \emph{blind FROG} \rev{or \emph{blind phaseless STFT}}. In this problem, the acquired data is the Fourier magnitude of  $y_{n,m}=x_n^1x_{n+mL}^2$ for  two signals $x^1,x^2\in\CN$ \cite{trebino2012frequency,wong2012simultaneously}. The goal is then to estimate both signals simultaneously from their phaseless measurements. \rev{Initial results on this model were derived in~\cite{bendory2017uniqueness}, yet we conjecture that they can be further improved by tighter analysis.}



\section*{References}

\bibliographystyle{elsarticle-num}
\bibliography{ref}

\end{document}